\documentclass[11pt,a4paper]{article}
\usepackage{amsfonts,amsmath,amssymb,amsthm}
\usepackage[margin=1in]{geometry}
\usepackage{array}
\usepackage{hyperref}
\usepackage{xcolor}
\usepackage[colorinlistoftodos,backgroundcolor=blue!10,linecolor=red,bordercolor=red]{todonotes}

\usepackage{lineno}

\newcolumntype{C}{>{\hfil$}p{\mylen}<{$\hfil}}

\newcommand{\heading}[1]{\medskip\par\noindent{\bf #1}}

\newenvironment{packed_itemize}{
    \begin{itemize}
        \setlength{\itemsep}{1pt}
        \setlength{\parskip}{0pt}
        \setlength{\parsep}{0pt}
}{\end{itemize}}

  \def\calF{{\mathcal F}}

  \def\calO{{\mathcal O}}
\def\calP{{\mathcal P}}

\def\cyc{\mathbb{Z}}
\def\real{\mathbb{R}}
\def\complex{\mathbb{C}}
\def\torus{\mathbb{T}}

\DeclareMathOperator{\diag}{diag}

\newtheorem{theorem}{Theorem}
\newtheorem{lemma}[theorem]{Lemma}

\theoremstyle{definition}

\newtheorem{example}{Example}

\title{Discrete and Fast Fourier Transform Made Clear}

\author{
	Peter Zeman
	\thanks{Department of Applied Mathematics, Faculty of Mathematics and Physics, Charles University, Prague, Czech Republic, \texttt{zeman@kam.mff.cuni.cz}.}
}

\date{}

\begin{document}

\maketitle

\begin{abstract}
Fast Fourier transform was included in the Top 10 Algorithms of 20th Century by \emph{Computing in Science \& Engineering}.
In this paper, we provide a new simple derivation of both the discrete Fourier transform and fast Fourier transform by means of elementary linear algebra.
We start the exposition by introducing the convolution product of vectors, represented by a circulant matrix, and derive the discrete Fourier transform as the change of basis matrix that diagonalizes the circulant matrix.
We also generalize our approach to derive the Fourier transform on any finite abelian group, where the case of Fourier transform on the Boolean cube is especially important for many applications in theoretical computer science.
\end{abstract}

\section{Introduction}

\emph{Discrete Fourier transform} is a change of basis matrix which for a vector given by the coordinates in the standard basis computes its coordinates in the Fourier basis.
Then \emph{fast Fourier transform} is just a fast way of computing the discrete Fourier transform of any vector.
Fast Fourier transform is according to \emph{Computing in Science \& Engineering} one of the top 10 most influential algorithms of 20th century~\cite{cipra2000best}.
Without a doubt, every theoretical computer scientist should have at least a basic understanding of the fundamental principles of this algorithm.

There are many excellent books dealing with these topics, for example~\cite{dasgupta2008algorithms,kantor2015mathematics++,pereyra2012harmonic,strang1986introduction,terras1999fourier}.
In this paper, we offer a new presentation of the discrete and the Fast Fourier transform.
We believe that the key in an exposition that makes the material more accessible to a wider audience, such as undergraduates and people interested in algorithms, is to start with some appealing motivation.
Our starting point the convolution product of two vectors in the complex vector space $\complex^n$.
In general, convolution is a fundamental concept in mathematics and it appears in many different forms with applications in image processing, digital data processing, acoustics, electrical engineering, physics, probability theory, multiplication of polynomials, and more.
An important part of our presentation is that no new notion is introduced out of the blue, but with a clear justification.

In the space $\complex^n$, the convolution $c * d$ of two vectors $c$ and $d$ is again a vector that can be represented by constructing the circulant matrix $C$ of $c$ and computing the product $Cd$.
It is clear that $c * d$ can be computed in $\calO(n^2)$ time, however, due to the numerous applications of convolution, it is desirable to compute it faster.
This naturally leads to the spectral decomposition of $C$, which is guaranteed by the very specific structure of a circulant matrix.
The basis of orthogonal eigenvectors which diagonalizes $C$ is exactly the Fourier basis.
Fast Fourier transform computes the coordinates of a vector in the Fourier basis in time $\calO(n\log{n})$, which then also gives an algorithm computing the convolution of two vectors in the same time.

In general, Fourier analysis provides an orthogonal basis of complex functions $G\to \complex$ defined on an abelian group $G$.
Note that the vectors in $\complex^n$ can be identified with the functions $\cyc_n \to \complex$, where $\cyc_n$ are the integers modulo $n$ with addition.
Some of the important settings are when $G$ for example $\real$, $\real/\cyc$, which correspond to the classical Fourier transform and Fourier series, respectively, $\cyc_n$, which corresponds to the discrete Fourier transform, and $\cyc_2^n$, which corresponds to the Fourier analysis on the Boolean cube.

In this paper, we focus on the case when $G$ is finite where we can think of the functions $G\to \complex$ as complex vectors indexed by $G$.
We generalize the approach for $G = \cyc_n$ to derive the Fourier transform on $G$.
To this end we recall the, not so well-known, definition of a $G$-circulant matrix, which can be used to represent the convolution of complex functions on $G$.
The Fourier basis is formed by the eigenvectors of the $G$-circulant.
We derive a recursive description of $G$-circulants and of their eigenvectors, which, to the best of our knowledge, is not stated in the literature in the form as it is in this paper.


The main idea underlying Fourier analysis on finite abelian groups is a basic fact of linear algebra: if a linear mapping has a orthogonal basis of eigenvectors, we can see it as a diagonal matrix in this basis.
The infinite-dimensional case is more complicated but the rough idea is similar.
We will see this manifested throughout the paper.

\heading{Prerequisites.}
We assume that the reader is familiar with basic linear algebra.
In particular, the important concepts which we need are the following: representation of a linear mapping by a matrix, eigenvalues and eigenvectors, spectral decomposition.

\heading{Structure of the paper.}
In Section~\ref{sec:dft}, we derive the discrete Fourier transform.
In Section~\ref{sec:fft}, we derive the Fast Fourier transform using a matrix notation, which is more transparent.
In Section~\ref{sec:abelian}, we generalize the approach from Section~\ref{sec:dft} to any finite abelian group $G$.

\heading{Acknowledgements.}
We would like to thank Pavel Klav\'{i}k, Martin \v{C}ern\'{y}, Milan Hlad\'{i}k, and Roman Nedela for many useful comments.

\section{Discrete Fourier Transform}
\label{sec:dft}

Our starting point is the discrete circular convolution, which naturally leads to the discrete Fourier transform.
It is an operation that, given two vectors $c, d \in \complex^n$, produces a third vector $c * d \in \complex^n$.
Instead of saying what each component of the vector $c * d$ looks like, we take a different approach.
For the vector $c$, we construct the $n\times n$ \emph{circulant matrix}:
$$
C = 
\left [\begin{array}{ccccc}
c_0		& c_{n-1}	& \cdot		& c_{2}		& c_{1}  \\
c_{1}		& c_0 		& c_{n-1}	& \cdot 	& c_{2}  \\
\cdot  		& c_{1}		& c_0    	& \cdot 	& \cdot   \\
c_{n-2}		& \cdot		& \cdot		& \cdot 	& c_{n-1}   \\
c_{n-1}		& c_{n-2}	& \cdot		& c_{1} 	& c_0
\end{array}\right ].
$$
The first column of $C$ is the vector $c$ and each column is the cyclic shift of the previous column by one position in the downward direction.

\begin{example}
\label{eq:circulant}
The circulant matrix corresponding to the vector
$\left [\begin{array}{ccc}
1 & 2 & 3
\end{array}\right ]^T$ is the matrix
$$
\left [\arraycolsep=3pt
\begin{array}{ccc}
1 & 3 & 2\\
2 & 1 & 3\\
3 & 2 & 1
\end{array}\right ].
$$
\qed
\end{example}

We define the \emph{convolution} of the vectors $c$ and $d$ by
$$c * d := Cd.$$
Here, the matrix $C$ represents the linear mapping $\complex^n \to \complex^n$ defined by $d \mapsto c * d$.

By definition, computing $c * d$ requires $n^2$ arithmetic operations, which is the number of operations needed for a matrix-vector multiplication.
However, one can see that the structure of the matrix $C$ is very special.
It turns out that the eigenvectors of $C$ form an orthogonal basis of $\complex^n$, which means that, in a suitable basis, the convolution can be represented by a diagonal matrix.
This can be used to derive an algorithm, called \emph{fast Fourier transform}, that computes $c * d$ in time $\calO(n\log{n})$.

\subsection{Eigenvectors of a Circulant Matrix}

In linear algebra, the first thing to do when one encounters a new linear mapping, is to try to compute its eigenvalues and eigenvectors.
In case the eigenvectors form an orthogonal basis, then the linear mapping is represented by a diagonal matrix with respect to this basis.
As we will see, this is exactly the case for every circulant.

Note that we can write
$$C = c_0I + c_{n-1}P + c_{n-2}P^2 + \cdots + c_1P^{n-1},$$
where
$$P =
\left [\arraycolsep=3pt
\begin{array}{ccccc}
0 & 1 & \cdot & 0 & 0 \\
0 & 0 & 1 & \cdot & 0 \\
\cdot  & 0 & 0 & \cdot  & \cdot \\
0 & \cdot & \cdot & \cdot & 1 \\
1 & 0 & \cdot & 0 & 0 \\
\end{array}\right ].
$$

\begin{example}
\label{ex:circulant_split}
For the vector
$\left [\begin{array}{ccc}
1 & 2 & 3
\end{array}\right ]^*$, we have
\begin{align*}
\left [\arraycolsep=3pt
\begin{array}{ccc}
1 & 3 & 2\\
2 & 1 & 3\\
3 & 2 & 1
\end{array}\right ]
&=
1\cdot
\left [\arraycolsep=3pt
\begin{array}{ccc}
1 &  & \\
 & 1 & \\
 &  & 1
\end{array}\right ]
+
3\cdot
\left [\arraycolsep=3pt
\begin{array}{ccc}
 & 1 & \\
 &  & 1\\
1 &  & 
\end{array}\right ]
+
2\cdot
\left [\arraycolsep=3pt
\begin{array}{ccc}
 &  & 1\\
1 &  & \\
 & 1 & 
\end{array}\right ]\\
&=
1\cdot
\left [\arraycolsep=3pt
\begin{array}{ccc}
 & 1 & \\
 &  & 1\\
1 &  & 
\end{array}\right ]^0
+
3\cdot
\left [\arraycolsep=3pt
\begin{array}{ccc}
 & 1 & \\
 &  & 1\\
1 &  & 
\end{array}\right ]^1
+
2\cdot
\left [\arraycolsep=3pt
\begin{array}{ccc}
 & 1 & \\
 &  & 1\\
1 &  & 
\end{array}\right ]^2.\\[2mm]
\end{align*}
\qed
\end{example}

If $v$ is an eigenvector of $P$, i.e., $P v = \lambda v$, then we have
\begin{align*}
C v &= (c_0I + c_{n-1}P + c_{n-2}P^2 + \cdots + c_1P^{n-1})v = c_0I v + c_{n-1}P v + c_{n-2}P^2 v + \cdots + c_1P^{n-1}v\\
&= c_0I v + c_{n-1}\lambda v + c_{n-2}\lambda^2 v + \cdots + c_1\lambda^{n-1}v = (c_0 + c_{n-1}\lambda + c_{n-2}\lambda^2 + \cdots + c_1\lambda^{n-1})v.
\end{align*}
In other words, every eigenvector $v$ of $P$ is an eigenvector of $C$.
The corresponding eigenvalue of $C$ can then be computed from the formula
\begin{equation}
\label{eq:circulant_eigenvalue1}
c_0 + c_{n-1}\lambda + c_{n-2}\lambda^2 + \cdots + c_1\lambda^{n-1},
\end{equation}
where $\lambda$ is the eigenvalue of $P$ corresponding to the eigenvector $v$.
To determine the eigenvalues and eigenvectors of $C$, it suffices to determine the eigenvalues and eigenvectors of $P$.

\heading{Eigenvectors of the Matrix $P$.}
Suppose that $P v = \lambda v$.
The matrix $P$ is a permutation matrix, and therefore, it is unitary, i.e., $P^*P = I$.
The eigenvalues of $P$ lie on the unit circle in $\complex$.
In fact, this is true for every unitary matrix $U$:
$$
\|v\| = \|U v\| = \|\lambda v\| = |\lambda|\cdot \|v\|.
$$
This implies $|\lambda| = 1$.
Moreover, we can see that the order of the permutation represented by $P$ is $n$, i.e., $P^n = I$.
It follows that if $\lambda$ is an eigenvalue of $P$, then $\lambda^n$ is an eigenvalue of $I$.
Since we know that $I$ has all its eigenvalues equal to $1$, the candidates for eigenvalues of $P$ are exactly the $n$-th roots of unity:
$w^k$, for $k = 0, 1, \dots, n-1,$ and $w = e^{2\pi i/n}$.

Now, we find the eigenvector
$\chi_k =
\left [\begin{array}{ccccc}
x_0 & x_1 & \cdot & x_{n-1} & x_{n-1}
\end{array}\right ]^T
$
associated to the eigenvalue $w^k$.
Suppose that $P \chi_k = w^k \chi_k$.
We have
$$
\left [\begin{array}{c}
x_{1}\\
x_{2}\\
\cdot\\
x_{n-1}\\
x_{0}
\end{array}\right ]
= P \chi_k = w^k\chi_k =
w^k
\left [\begin{array}{c}
x_{0}\\
x_{1}\\
\cdot\\
x_{n-2}\\
x_{n-1}
\end{array}\right ].
$$
The entries $x_0,\dots,x_{n-1}$ of the vector $\chi_k$ must satisfy the following system of linear equations:
\begin{align*}
x_{1} &= w^k x_0,\\
x_{2} &= w^kx_{1} = w^{2k} x_0,\\
&\vdots\\
x_{n-1} &= w^kx_{n-2} = \cdots = w^{(n-1)k}x_0,\\
x_0 &= w^kx_{n-1} = \cdots = w^{nk}x_0 = x_0.
\end{align*}
If we pick an arbitrary value for $x_0$, then $x_1,\dots,x_{n-1}$ are uniquely determined.
We put $x_0 = w^{0k} = 1$.
The eigenvectors of $P$, and therefore also the eigenvectors of $C$, are
$$
\chi_k =
\left [\def\arraystretch{1.5}
\begin{array}{c}
w^{0\cdot k}\\
w^{1\cdot k}\\
\cdot\\
w^{(n-1)k}
\end{array}\right ]
=
\left [\def\arraystretch{1.5}
\begin{array}{c}
e^{2\pi i0k/n}\\
e^{2\pi i1k/n}\\
\cdot\\
e^{2\pi i(n-1)k/n}
\end{array}\right ],
$$
for $k = 0,\dots,n-1$.
The corresponding eigenvalues of $C$, satisfying $C\chi_k = \Lambda_k\chi_k$, are given by
$$
\Lambda_k = c_0 + c_{n-1}w^k + c_{n-2}w^{2k} + \cdots + c_1w^{(n-1)k}.
$$

\subsection{Fourier Basis}

The key property of the eigenvectors $\chi_0,\dots,\chi_{n-1}$ of $C$ is that they form an
orthogonal basis of $\complex^n$, called the \emph{Fourier basis}. This means that we can apply spectral
decomposition to every circulant matrix $C$.

\begin{lemma}
\label{lem:fourier_orthogonal}
The vectors $\chi_0,\dots,\chi_{n-1}$ form an orthogonal basis.
\end{lemma}

\begin{proof}
We have
$$\chi_k^*\chi_\ell = \sum_{j=0}^{n-1} \overline{w^{jk}}w^{j\ell} =
\sum_{j=0}^{n-1} w^{-jk}w^{j\ell} = \sum_{j=0}^{n-1} w^{j(\ell-k)}.$$
If $k = \ell$, then clearly $\chi_k^*\chi_\ell = n$. If $k \neq \ell$, then we have
$$\chi_k^*\chi_\ell = \sum_{j=0}^{n-1} w^{(\ell - k)j} = \alpha^0 + \alpha^1 + \cdots +
\alpha^{n-1} = \frac{1 - \alpha^n}{1 - \alpha} = 0,$$
where $\alpha = w^{\ell- k}$.  The last equality follows from the fact that
$e^{2\pi ik} = 1$.

There is more geometric way of proving that
$\alpha^0 + \alpha^1 + \cdots + \alpha^{n-2} + \alpha^{n-1} = 0$.
Since the complex number $\alpha$ is nonzero, $\alpha z = z$ implies $z = 0$.
The equality $\alpha z = z$ is clearly satisfied by $\alpha^0 + \alpha^1\cdots+ \alpha^{n-2}+\alpha^{n-1}$.
Geometrically, multiplication by $\alpha$ is a rotation by some angle $\theta$.
The complex numbers $\alpha^0, \alpha^1, \dots, \alpha^{n-1}$ are the vertices of a regular $m$-gon (note that $m$ divides $n$).
The angle $\theta$ is then equal to $2\pi/m$.
Rotation by angle $\theta$ just permutes the vertices of the $m$-gon, thus the sum remains the same.
\end{proof}

\subsection{Discrete Fourier Transform as Change of Basis Matrix}

Discrete Fourier transform is the change of basis matrix from the standard basis to the Fourier basis $\{\chi_0,\dots,\chi_{n-1}\}$.
We can easily find the change of basis matrix $F$ from the Fourier basis to the standard basis, by placing the vectors $\chi_0,\dots,\chi_{n-1}$ into the columns:
$$
F =
\left [\begin{array}{cccc}
\vline & \vline & & \vline\\ 
\chi_0 & \chi_1 & \cdot & \chi_{n-1}\\
\vline & \vline & & \vline\\ 
\end{array}\right ].
$$
Note that in Lemma~\ref{lem:fourier_orthogonal}, we proved that
$$
F^*F =
\left [\begin{array}{ccc}
\rule[.5ex]{0.5cm}{0.4pt} & \chi_0^* & \rule[.5ex]{0.5cm}{0.4pt}\\
\rule[.5ex]{0.5cm}{0.4pt} & \chi_1^* & \rule[.5ex]{0.5cm}{0.4pt}\\
& \cdot &\\
\rule[.5ex]{0.5cm}{0.4pt} & \chi_{n-1}^* & \rule[.5ex]{0.5cm}{0.4pt}\\
\end{array}\right ]\cdot
\left [\begin{array}{cccc}
\vline & \vline & & \vline\\ 
\chi_0 & \chi_1 & \cdot & \chi_{n-1}\\
\vline & \vline & & \vline\\ 
\end{array}\right ]
= nI.
$$
By rearranging, we have $\frac{1}{n}F^*F = I$. The matrix $D:= \frac{1}{n}F^*$ is
called \emph{discrete Fourier transform} and the \emph{inverse discrete Fourier transform} is $D^{-1} = \left (\frac{1}{n}F^*\right )^{-1} =
F$. In particular, we have
$$F =
\left [\begin{array}{ccccc}
1 & 1 & 1 & \cdot & 1\\
1 & w & w^2 & \cdot & w^{n-1}\\
1 & w^2 & w^4 & \cdot & w^{2(n-1)}\\
\cdot & \cdot & \cdot & \cdot & \cdot\\
1 & w^{n-1} & w^{2(n-1)} & \cdot & w^{(n-1)^2}\\
\end{array}\right ]
$$
and
$$
F^{-1} = \frac{1}{n}
\left [\begin{array}{ccccc}
1 & 1 & 1 & \cdot & 1\\
1 & w & w^{-2} & \cdot & w^{-(n-1)}\\
1 & w^{-2} & w^{-4} & \cdot & w^{-2(n-1)}\\
\cdot & \cdot & \cdot & \cdot & \cdot\\
1 & w^{-(n-1)} & w^{-2(n-1)} & \cdot & w^{-(n-1)^2}\\
\end{array}\right ].
$$
Actually, the matrix $F$ is symmetric, so the inverse can be computed by merely
taking the complex conjugate of $F$ and dividing by $n$, i.e., $F^{-1} =
\frac{1}{n}F^* = \frac{1}{n}\overline{F}$.


\subsection{Convolution Theorem}

We use the matrix $F$ to diagonalize $C$.
Since the columns of the matrix $F$ are exactly the eigenvectors of $C$, we have the following
\begin{equation}
\label{eq:diag1}
C = D^{-1}\diag(\Lambda_0,\dots,\Lambda_{n-1})D,
\end{equation}
where $\Lambda_{k}$ is the eigenvalue of $C$ corresponding to the eigenvector $\chi_k$.
The equation~\ref{eq:diag1} can be interpreted as follows: to apply the linear mapping $C$ it is the same as to change to the Fourier basis using $D$, then to apply a diagonal matrix, and then to change back to the standard basis using $D^{-1}$.
Note that spectral decomposition in fact requires orthonormal basis, i.e., every vector must have norm $1$, but in our case the Fourier basis orthogonal since all the vectors are of norm $n$.
This can be easily modified, but we stick here to the notation that is usually used in the textbooks.

From the previous analysis we know that $P \chi_k = w^k \chi_k$.
Moreover, from~(\ref{eq:circulant_eigenvalue1}) it follows that
\begin{equation}
\label{eq:circulant_eigenvalue2}
\Lambda_k = c_0w^{0k} + c_{n-1}w^{1k} + c_{n-1}w^{2k} + \cdots + c_1w^{(n-1)k}.
\end{equation}
The right-hand side can be rearranged using the following relations:
\begin{equation}
\label{eq:conjugate_relations}
w^{-1k} = w^{(n-1)k}, \quad w^{-2k} = w^{(n-2)k}, \quad \dots,\quad w^{-(n-1)k} = w^{1k}.
\end{equation}
By applying~(\ref{eq:conjugate_relations}) to~(\ref{eq:circulant_eigenvalue2}), we get
\begin{equation}
\label{eq:circulant_eigenvalue3}
\Lambda_k = c_0w^{0k} + c_{1}w^{-1k} + c_{2}w^{-2k} + \cdots + c_{n-1}w^{-(n-1)k}.
\end{equation}
Notice that the vector
$$
\left [\begin{array}{c}
w^{0k}\\
w^{-1k}\\
w^{-2k}\\
\cdot\\
w^{-(n-1)k}\\
\end{array}\right ]
$$
is exactly the $k$-th column of $\overline{F}$. Therefore,
\begin{equation}
\label{eq:circulant_eigenvalue4}
\overline{F}c = 
\left [\begin{array}{c}
\Lambda_{0}\\
\Lambda_{1}\\
\Lambda_{2}\\
\cdot\\
\Lambda_{n-1}
\end{array}\right ],
\end{equation}
i.e., the diagonal entries of $\diag(\Lambda_0,\dots,\Lambda_{n-1})$ are exactly the components of $\overline{F}c = nDc$.

From~(\ref{eq:circulant_eigenvalue4}), we get
\begin{equation}
\label{eq:diag2}
c*d = Cd = D^{-1}\diag(\Lambda_0,\dots,\Lambda_{n-1})Dd = D^{-1} (nDc \circ Dd),
\end{equation}
where $\circ$ denotes the \emph{Hadamard product} (componentwise) of two vectors.
We proved the following.

\begin{theorem}[Convolution theorem]
\label{thm:convolution_theorem}
For any two vectors $c,d \in \complex^n$, we have $$D(c * d) = n(Dc \circ Dd).$$
\end{theorem}

In words, the previous theorem states that the discrete Fourier transform of the convolution of two vectors equals (up to scaling by $n$) the componentwise product of the discrete Fourier transforms of $c$ and $d$.
It is possible to get rid of the scaling factor $n$ by choosing an orthonormal basis instead of orthogonal.

The complexity depends on how fast we can multiply a vector by the matrices $F$ and
$F^{-1}$. Using the very special structure of $F$ and $F^{-1}$, this can be done in
$\calO(n\log n)$, which is the main theme of the next section.

\section{Fast Fourier Transform}
\label{sec:fft}

We give the divide and conquer algorithm computing the discrete Fourier transform and the inverse discrete Fourier
transform in time $\calO(n\log n)$, which is due to Cooley and Tukey~\cite{cooley1965algorithm}.
For simplicity, we only show how to compute the product $Fx$, the other can be done analogously.
We start with the following key lemma.

\begin{lemma}
\label{lem:w_squared}
If $m = n/2$, then $w_n^2 = w_m$.
\end{lemma}
\begin{proof}
We have $w_n^2 = e^{\frac{4\pi i}{n}} = e^{\frac{2\pi i}{m}} = w_m$. The statement of the lemma has
the following geometric explanation. Let $\theta = \frac{2\pi}{n}$ be the angle of $w_n$. The
angle of $w_n^2$ is $2\theta$. On the other hand, $m$ divides $2\pi$ into less, but larger parts
than $w_n$. Since $m = n/2$, $\frac{2\pi}{m} = \frac{4\pi}{n} = 2\theta$.
\end{proof}

We show on an example for $n=4$ how to derive the recursive fast Fourier transform in a matrix notation.
Then we discuss how to generalize this for $n=2^k$, which we can always assume for simplicity.
In this chapter we use $F_n$ to indicate the fact that we are considering a matrix of size $n\times n$.
We also use $w_n$ to denote $e^{2\pi i/n}$.

We want to compute the product $y = F_4 x$.
We start just by applying the regular matrix-vector multiplication.
To find recursion, we apply Lemma~\ref{lem:w_squared} and the properties of complex conjugation.
We have
$$
y =
\left [\begin{array}{c}
y_0\\
y_1\\
y_2\\
y_3
\end{array}\right]
=
\left [\begin{array}{cccc}
1 & 1 & 1 & 1\\
1 & w_4 & w_4^2 & w_4^{3}\\
1 & w_4^2 & w_4^4 & w_4^{6}\\
1 & w_4^{3} & w_4^{6} & w_4^{9}
\end{array}\right ]
\cdot
\left [\begin{array}{c}
\textcolor{red}{x_0}\\
\textcolor{blue}{x_1}\\
\textcolor{red}{x_2}\\
\textcolor{blue}{x_3}
\end{array}\right ]
= F_4 x.
$$
We multiply and rearrange each component of the resulting vector:
$$
\left [\begin{array}{c}
y_0\\
y_1\\
y_2\\
y_3
\end{array}\right ]
=
\left [\begin{array}{c}
\textcolor{red}{x_0} + \textcolor{blue}{x_1} + \textcolor{red}{x_2} + \textcolor{blue}{x_3}\\
\textcolor{red}{x_0} + w_4\textcolor{blue}{x_1} + w_4^2\textcolor{red}{x_2} + w_4^{3}\textcolor{blue}{x_3}\\
\textcolor{red}{x_0} + w_4^2\textcolor{blue}{x_1} + w_4^4\textcolor{red}{x_2} + w_4^{6}\textcolor{blue}{x_3}\\
\textcolor{red}{x_0} + w_4^{3}\textcolor{blue}{x_1}  + w_4^{6}\textcolor{red}{x_2} + w_4^{9}\textcolor{blue}{x_3}
\end{array}\right ]
=
\left [\begin{array}{c}
\textcolor{red}{x_0} + \textcolor{red}{x_2} + (\textcolor{blue}{x_1} + \textcolor{blue}{x_3})\\
\textcolor{red}{x_0} + w_4^2\textcolor{red}{x_2} + w_4(\textcolor{blue}{x_1} + w_4^{2}\textcolor{blue}{x_3})\\
\textcolor{red}{x_0} + w_4^4\textcolor{red}{x_2} + w_4^2(\textcolor{blue}{x_1} + w_4^{4}\textcolor{blue}{x_3})\\
\textcolor{red}{x_0} + w_4^{6}\textcolor{red}{x_2}  + w_4^{3}(\textcolor{blue}{x_1} + w_4^{6}\textcolor{blue}{x_3})
\end{array}\right ].
$$
We apply Lemma~\ref{lem:w_squared} and use the properties of the complex numbers $w_2$ and $w_4$ to
obtain:
$$
\left [\begin{array}{c}
y_0\\
y_1\\
y_2\\
y_3
\end{array}\right ]
=
\left [\begin{array}{c}
\textcolor{red}{x_0} + \textcolor{red}{x_2} + (\textcolor{blue}{x_1} + \textcolor{blue}{x_3})\\
\textcolor{red}{x_0} + w_2\textcolor{red}{x_2} + w_4(\textcolor{blue}{x_1} + w_2\textcolor{blue}{x_3})\\
\textcolor{red}{x_0} + w_2^2\textcolor{red}{x_2} + w_4^2(\textcolor{blue}{x_1} + w_2^{2}\textcolor{blue}{x_3})\\
\textcolor{red}{x_0} + w_2^{3}\textcolor{red}{x_2}  + w_4^{3}(\textcolor{blue}{x_1}+ w_2^{3}\textcolor{blue}{x_3})
\end{array}\right ]
=
\left [\begin{array}{c}
\textcolor{red}{x_0} + \textcolor{red}{x_2} + (\textcolor{blue}{x_1} + \textcolor{blue}{x_3})\\
\textcolor{red}{x_0} + w_2\textcolor{red}{x_2} + w_4(\textcolor{blue}{x_1} + w_2\textcolor{blue}{x_3})\\
\textcolor{red}{x_0} + \textcolor{red}{x_2} - (\textcolor{blue}{x_1} + \textcolor{blue}{x_3})\\
\textcolor{red}{x_0} + w_2\textcolor{red}{x_2} - w_4(\textcolor{blue}{x_1}+ w_2\textcolor{blue}{x_3})
\end{array}\right ].
$$
In the last equation, we are using the fact that $w_4^0 = -w_4^2$ and $w_4^1 = -w_4^3$.
For general $n$, we have $w_n^k = -w_n^{k+n/2}$.
Now, we rearrange again:
$$
\left [\begin{array}{c}
y_0\\
y_1\\
y_2\\
y_3
\end{array}\right ]
=
\left [\begin{array}{c}
\left [\begin{array}{cc}
1 & 1\\
1 & w_2
\end{array}\right ]
\cdot
\left [\begin{array}{c}
\textcolor{red}{x_0}\\
\textcolor{red}{x_2}
\end{array}\right ] +
\left [\begin{array}{cc}
1 & \\
 & w_4
\end{array}\right ]
\cdot
\left [\begin{array}{cc}
1 & 1\\
1 & w_2
\end{array}\right ]
\cdot
\left [\begin{array}{c}
\textcolor{blue}{x_1}\\
\textcolor{blue}{x_3}
\end{array}\right ]\\[4mm]
\left [\begin{array}{cc}
1 & 1\\
1 & w_2
\end{array}\right ]
\cdot
\left [\begin{array}{c}
\textcolor{red}{x_0}\\
\textcolor{red}{x_2}
\end{array}\right ] +
\left [\begin{array}{cc}
-1 & \\
 & -w_4
\end{array}\right ]
\cdot
\left [\begin{array}{cc}
1 & 1\\
1 & w_2
\end{array}\right ]
\cdot
\left [\begin{array}{c}
\textcolor{blue}{x_1}\\
\textcolor{blue}{x_3}
\end{array}\right ]
\end{array}\right ].
$$
Finally, we get
$$
\left [\begin{array}{c}
y_0\\
y_1\\
y_2\\
y_3
\end{array}\right ]
=
\left [\begin{array}{cccc}
1 &  & 1 & \\
 & 1 &  & w_4\\
1 &  & -1 & \\
 & 1 &  & -w_4
\end{array}\right ]
\cdot
\left [\begin{array}{cccc}
1 & 1 &  & \\
1 & w_2 &  & \\
 &  & 1 & 1\\
 &  & 1 & w_2
\end{array}\right ]
\cdot
\left [\begin{array}{c}
\textcolor{red}{x_0}\\
\textcolor{red}{x_2}\\
\textcolor{blue}{x_1}\\
\textcolor{blue}{x_3}
\end{array}\right ].
$$
We put
$$
\left [\begin{array}{c}
y_0\\
y_1\\
y_2\\
y_3
\end{array}\right ]
=
\left [\begin{array}{cc}
I_2 & A_2\\
I_2 & -A_2
\end{array}\right ]
\cdot
\left [\begin{array}{cc}
F_2 & \\
 & F_2
\end{array}\right ]
\cdot
{
\arraycolsep=3pt
\left [\begin{array}{cccc}
1 &  &  & \\
 &  & 1  & \\
 & 1 &  & \\
 &  &  & 1
\end{array}\right ]
}
\cdot
\left [\begin{array}{c}
\textcolor{red}{x_0}\\
\textcolor{blue}{x_1}\\
\textcolor{red}{x_2}\\
\textcolor{blue}{x_3}
\end{array}\right ].
$$

In general, for $n = 2^k$, by a similar derivation, we obtain:
$$
y =
F_nx
=
\left [\begin{array}{cc}
I_{n/2} & A_{n/2}\\
I_{n/2} & -A_{n/2}
\end{array}\right ]
\cdot
\left [\begin{array}{cc}
F_{n/2} & \\
 & F_{n/2}
\end{array}\right ]
\cdot
P_\pi \cdot x,
$$
where $I_{n/2}$ is the $n/2\times n/2$ identity matrix, $A_{n/2} = \diag(1,w_{n},\dots,w_{n}^{n/2-1})$, and $\calP_\pi$ is the permutation matrix which put first all the even components of $x$.

Let $T(n)$ be the time needed to compute $F_n x$. First, the algorithm splits the vector $x$
into the even-numbered components $x_\text{even}$ and the odd-numbered components $x_\text{odd}$,
which can be done in time $\Theta(n)$. Then, the algorithm computes
$F_{\frac{n}{2}}x_\text{even}$ and $F_{\frac{n}{2}}x_\text{even}$. This takes $T(n/2)$ time.
The final multiplication can be also done in the $\Theta(n)$. We get the following recurrence:
$$T(n) = 2T(n/2) + \Theta(n).$$
Using standard methods, it can be easily shown that $T(n) = \calO(n\log{n})$.

\section{Discrete Fourier Transform on Finite Abelian Groups}
\label{sec:abelian}

In this section, we consider functions of the form $v\colon G\to \complex$, where $G$ is a finite abelian group.
We denote the set of all such functions by $\complex^G$.
The set $\complex^G$ is clearly a vector space.
In this section, we use the following shorthand $e(x) := e^{2\pi ix}$.

It is a folklore that every finite abelian group is isomorphic to a direct product of the form
$$\cyc_{k_1}\times \cdots\times \cyc_{k_u},$$
where $k_1,\dots,k_u$ are powers of (\emph{not necessarily distinct}) primes; see for example~\cite{rotman2012introduction}.
For every finite abelian group, we fix a canonical form by choosing the sequence $k_1,\dots,k_u$ to be non-decreasing.
In what follows, when talking about a finite abelian group $G$, we always think about this canonical representation instead.

We order the elements of an abelian group $G$ lexicographically.
Then, we can alternatively think of $v\in\complex^G$ as a complex vector $(v_x)_{x\in G}$ indexed by the group $G$.
We refer to the elements of $\complex^G$ as vectors.
However, we prefer the functional notation for cosmetic reasons.

\heading{Convolution.}
The previous two sections dealt with the case when $G = \cyc_n$.
Here we derive the Fourier transform on $G$, for any finite abelian group $G$.
We again start with the important convolution product of two vectors $v,u\in\complex^G$, which is usually defined by
$$(v * u)(x)  = \sum_{y\in G}v(x - y)g(y).$$
Similarly, the concept of circulant matrix can be generalized for finite abelian groups.
For $v \in \complex^G$, the \emph{$G$-circulant matrix} $C$, indexed by $G$, is a $|G|\times |G|$ matrix defined by
$$C(x,y) := v(x - y),$$
for $x,y \in G$.
Note that here $C(x,y)$ denotes the entry of $C$ indexed by $x$ and $y$.
We use this notation throughout the section.

\begin{example}
Let $G = \cyc_3\times\cyc_2$ and let $v \in \complex^G$ be a vector such that $v = \left[\begin{array}{cccccc}a & b & c & d & e & f\end{array}\right]^T$.
Then the $G$-circulant $C$ of $v$ is the following matrix:
$$
\left[
\begin{array}{cccccc}
a & b & e & f & c & d\\
b & a & f & e & d & c\\
c & d & a & b & e & f\\
d & c & b & a & f & e\\
e & f & c & d & a & b\\
f & e & d & c & b & a
\end{array}
\right].
$$
\qed
\end{example}

The convolution $v*u$ can be represented by multiplying the vector $u$ by the matrix $C$ from the left.

\begin{lemma}
For $v,u \in \complex^G$, we have $v * u  = Cu$, where $C$ is the $G$-circulant matrix of $v$.
\end{lemma}

\begin{proof}
We have $(Cu)(x) = \sum_{y\in G} v(x-y)u(y) = (v*u)(x)$.
\end{proof}

\heading{Eigenvectors of $G$-circulants.}
The following theorem gives a recursive description of any $G$-circulant matrix and of its eigenvectors.
Although it is not difficult to prove, to the best of our knowledge, this theorem is not available in this form in the current literature.

Before stating the theorem, we recall that the \emph{Kronecker product} of an $m\times n$ matrix $A$ with an $p\times q$ matrix $B$ is a $mp\times nq$ matrix $A\otimes B$ defined by
$$
A\otimes B =
\arraycolsep=3pt
\left[
\begin{array}{ccc}
a_{1,1}B & \cdots & a_{1,n}B \\
\vdots & \ddots & \vdots \\
a_{m,1}B & \cdots & a_{m,n}B \\
\end{array}
\right].
$$
We also need several properties of the Kronecker product:
\begin{packed_itemize}
\item
If $Av = \lambda v$ and $Bu = \mu u$, then $\lambda\mu$ is an eigenvalue of $A\otimes B$ with the corresponding eigenvector $v\otimes u$.
Any eigenvalue of $A\otimes B$ arises as such product of eigenvalues of $A$ and $B$.
\item
We have $(A\otimes B)^* = A^* \otimes B^*$.
\item
If $A$ and $C$ are of the same size, $B$ and $D$ are of the same size, then $(A\otimes B) \circ (C\otimes D) = (A\circ C)\otimes (B\circ D)$, where $\circ$ denotes the Hadamard (componentwise) product.
\end{packed_itemize}

\begin{theorem}
\label{thm:g_circulant_recursive}
Let $G$ and $G'$ be abelian groups such that $G$ is not cyclic and $G = \cyc_k \times G'$.
Then for every $G$-circulant $C$, the following hold:
\begin{itemize}
\item[(i)]
There are $G'$-circulants $C_0,\dots,C_{k-1}$ such that
$$C = I\otimes C_0 + P \otimes C_{k-1} + P^2 \otimes C_{k-2} + \cdots + P^{k-1}\otimes C_1,$$
where $P$ is defined in Section~\ref{sec:dft}.
\item[(ii)]
For every $g = (g_1,g') \in G$, there is an eigenvector $\chi_g$ of $C$ such that $\chi_g = \chi_{g_1} \otimes \chi_{g'}$, where $\chi_{g_1}$ is an eigenvector of any $\cyc_k$-circulant, and $\chi_{g'}$ is an eigenvector of any $G'$-circulant.
In particular every $G$-circulant has the same set of eigenvectors.
\end{itemize}
\end{theorem}

\begin{proof}
First, we prove (i).
Consider the submatrix $C_{i,j}$ of $C$ indexed by the sets $I = \{(i,x') \in G : x' \in G'\}$ and $J = \{(j,x') \in G : x' \in G'\}$, for some fixed $i,j \in \cyc_k$.
For every $x \in I$ and $y \in J$, we have
$$C_{i,j}(x,y) = f(x - y) = f((i-j,x'-y')).$$
It follows that $C_{i,j}$ is a $G'$-circulant.

Clearly, $C_{i,j} = C_{i+m,j+m}$, for every $m \in \cyc_k$.
The mapping $(i,j) \mapsto (i+m,j+m)$ defines an action of $\cyc_k$ on $\cyc_k^2$.
We have $C_{i,j} = C_{i',j'}$ if $(i,j)$ and $(i',j')$ belong to the same orbit under this action.
The powers of the permutation matrix $P$ are in one-to-one correspondence with the orbits of this action.
We put $C_m := C_{0,k - m}$, for $m = 0,\dots, k-1$, and the addition taken modulo $k$.
This concludes the proof of part (i).

Now, we prove (ii).
From (i), we know that if $v$ is an eigenvector of $C$, then it is also an eigenvector of each $P^i\otimes C_{k-i}$, for $i=0,\dots,k-1$.
By induction, we may assume that every $G'$-circulant has the same set of eigenvectors.
Therefore, $v$ must be of the form $\chi_{g_1}\otimes \chi_{g'}$, where $\chi_{g_1}$ and $\chi_{g'}$ are exactly as in the statement.
\end{proof}

\heading{Fourier Basis.}
Denote the set of eigenvectors of any $G$-circulant by $\calF = \{\chi_g : g \in G\}$.
Before we prove that $\calF$ forms an orthogonal basis of $\complex^{G}$, we derive a precise description of every $\chi_g \in\calF$.

\begin{lemma}
\label{lem:eigenvectors}
Let $G = \cyc_{k_1}\times \cdots \times\cyc_{k_u}$ be an abelian group.
Then for $\chi_g \in \calF$,
$$\chi_g(x) = e\left( \sum_{i=1}^u\frac{g_i x_i}{k_i}\right),\quad x\in G.$$
\end{lemma}

\begin{proof}
We prove the lemma by induction on $u$.
If $u = 1$, then the group $G$ is cyclic and the statement follow from the arguments in Section~\ref{sec:dft}.
Suppose that the statement holds for any $G'$ with at most $u-1$ factors.
By Theorem~\ref{thm:g_circulant_recursive}, we have $\chi_g = \chi_{g_1} \otimes \chi_{g'}$, where $g = (g_1,g')$, for some $g_1 \in \cyc_{k_1}$ and $g' \in G' = \cyc_{k_2}\times\cdots\times \cyc_{k_u}$.
By the induction hypothesis and the definition of the Kronecker product we have
$$\chi_g(x) = (\chi_{g_1} \otimes \chi_{g'})(x) = \chi_{g_1}(x_1)\chi_{g'}(x_2,\dots,x_u) =$$
$$= e\left(\frac{g_1 x_1}{k_1}\right) \cdot e\left( \sum_{i=2}^u\frac{g_i x_i}{k_i}\right) = e\left( \sum_{i=1}^u\frac{g_i x_i}{k_i}\right).$$
\end{proof}

Recall that the inner product on $\complex^G$ is defined in an expected way
$v^*u := \sum_{x\in G}\overline{v(x)}u(x)$.
Finally, we prove that the set $\calF$ forms an orthogonal basis, called the \emph{Fourier basis}.

\begin{theorem}[Fourier basis]
The set $\calF = \{\chi_g : g \in G\}$ forms an orthogonal basis.
\end{theorem}

\begin{proof}
In Section~\ref{sec:dft}, we proved this for the case when $G$ a cyclic group.
Assume now that $G$ is not cyclic.
Let $\chi_g, \chi_h \in \calF$.
By Theorem~\ref{thm:g_circulant_recursive}, we have
$\chi_g = \chi_{g_1}\otimes\chi_{g'}$ and $\chi_h = \chi_{h_1}\otimes\chi_{h'}$.
Moreover, from the properties of the Kronecker product, it follows that
$$\overline{\chi_g} \circ \chi_h = (\overline{\chi_{g_1}}\otimes\overline{\chi_{g'}})\circ (\chi_{h_1}\otimes\chi_{h'}) = (\overline{\chi_{g_1}}\circ\chi_{h_1})\otimes(\overline{\chi_{g'}}\circ\chi_{h'}).$$
To compute $\chi_g^*\chi_h$, we proceed by induction. It suffices to evaluate the sum
$$\sum_{x\in G}(\overline{\chi_g}\circ\chi_h)(x) = \sum_{(x_1,x')\in G}\overline{\chi_{g_1}(x_1)}\cdot\chi_{h_1}(x_1)\cdot\overline{\chi_{g'}(x')}\cdot\chi_{h'}(x') =$$
$$= \sum_{x_1\in\cyc_{k_1}}\left(\overline{\chi_{g_1}(x_1)}\cdot\chi_{h_1}(x_1)\sum_{x'\in G'}\overline{\chi_{g'}(x')}\cdot\chi_{h'}(x')\right)
= \sum_{x_1\in\cyc_{k_1}}\overline{\chi_{g_1}(x_1)}\cdot\chi_{h_1}(x_1)\cdot \chi_{g'}^*\chi_{h'}.$$
If $g = h$, then by it easily follows from Lemma~\ref{lem:eigenvectors} that $\chi_g^*\chi_h = |G|$.
If $g\neq h$, then by induction, we have $\chi_{g'}^*\chi_{h'} = 0$, therefore, also $\chi_{g}^*\chi_{h} = 0$.
\end{proof}

The standard basis of $\complex^G$ consists of vectors $\delta_g$, where $\delta_g(x) = 1$ if $x = g$ and $\delta_g(x) = 0$ if $x\neq g$.
The \emph{discrete Fourier transform on the finite abelian group $G$} is then the change of basis matrix from the standard basis to the Fourier basis.
The matrix could be derived in an analogous way as for $G = \cyc_n$ in Section~\ref{sec:dft}, but we omit it here.

\heading{Fourier Analysis on Boolean Cube.}
If $G = \cyc_2^n$, then $G$ is often called \emph{Boolean cube}.
We note that in this special case the eigenvectors can be identified with the subsets of $\{1,\dots,n\}$.
In particular, for every $g\in \cyc_2^n$, we put $S_g := \{i : g_i = 1\}$.
Then by Lemma~\ref{lem:eigenvectors}, for $x \in \cyc_2^n$, we have 
$$\chi_g(x) = (-1)^{\sum_{i\in S_g}x_i}.$$
Fourier analysis on Boolean cube is especially important in many applications in theoretical computer science; for a survey see~\cite{de2008brief}.

\heading{Characters of an Abelian Group.}
Finally, we conclude this section by comparing our derivation of Fourier transform on $G$ with the approach in other textbooks. 
Typically, the exposition starts by introducing characters of abelian groups.

Let $\torus = \{z\in\complex : |z| = 1\}$ be the multiplicative group of complex numbers on the unit circle.
A \emph{character}, also called \emph{one-dimensional representation}, of a finite abelian group $G$ is a homomorphism $\chi\colon G\to\torus$.
It can be proved that characters form an orthogonal basis and this is then called the Fourier basis.
The issue is that at the first encounter, this construction, though elegant, might seem to come out of the blue.

On the other hand, we first start with the convolution product.
Its great importance clearly motivates our next endeavour.
We realize that it can by represented by a linear mapping, namely the $G$-circulant.
Then the natural step is to compute its eigenvectors and to discover that they form an orthogonal basis.
Moreover, from Lemma~\ref{lem:eigenvectors} it is also easy to see that every eigenvector of a $G$-circulant is a character, and it is also straightforward to prove that every character of an abelian group $G$ is an eigenvector of every $G$-circulant.

Nevertheless, the connection between characters and Fourier basis is crucial when developing Fourier analysis on finite non-abelian groups, which relies on representation theory of groups.
This is out of the scope of this paper.

\bibliographystyle{plain}
\bibliography{fourier}

\end{document}